%% file: main.tex
\documentclass[11pt]{article}
\usepackage[letterpaper,margin=1in]{geometry}

\usepackage{enumitem}
\input{macros.tex}

\title{A Topological Version of Schaefer's Dichotomy Theorem}

\let\anonymous\undefined %

\ifdefined\anonymous
    \author{Redacted}
    \affil{Affiliation}
    \usepackage[]{lineno}
    \linenumbers
\else
\author[1]{Patrick Schnider}
\author[1]{Simon Weber}
\affil[1]{Department of Computer Science, ETH Z\"{u}rich, Switzerland}
\fi

\date{}

\begin{document}

\maketitle

\begin{abstract}
Schaefer's dichotomy theorem [Schaefer, STOC'78] states that a boolean constraint satisfaction problem (CSP) is polynomial-time solvable if one of six given conditions holds for every type of constraint allowed in its instances. Otherwise, it is \NP-complete. In this paper, we analyze boolean CSPs in terms of their \emph{topological complexity}, instead of their \emph{computational complexity}. We attach a natural topological space to the set of solutions of a boolean CSP and introduce the notion of \emph{projection-universality}.
We prove that a boolean CSP is projection-universal if and only if it is categorized as $\mathsf{NP}$-complete by Schaefer's dichotomy theorem, showing that the dichotomy translates exactly from computational to topological complexity. We show a similar dichotomy for SAT variants and homotopy-universality.
\end{abstract}

{\small
\ifdefined\anonymous
\vfill
\else
\paragraph{Acknowledgments.} We thank Tillmann Miltzow for the helpful discussion. Simon Weber is supported by the Swiss National Science Foundation under project no. 204320.
\fi
}

\ifdefined\anonymous
\clearpage
\fi

\section{Introduction}

In this paper we study properties of solution spaces and their relation to complexity theory. Such questions originated in the following geometric problem: Given a point set $P$, its \emph{order type} is a combinatorial description of $P$, assigning each triple of points $(p,q,r)$ its orientation. The problem of \emph{order type realizability} asks whether for a given abstract order type there exists a point set with this order type. In 1956, Ringel~\cite{Ringel1956} asked whether for every realizable order type, every point set~$P$ with this order type can be continuously transformed into any other point set $P'$ with the same order type. 
Considering the uncountably infinite family of such point sets as a \emph{topological space}, Ringel's question is equivalent to asking whether this space is connected. We also call this space the \emph{realization space} of the given abstract order type, or alternatively, the \emph{solution space} of the given order type realizability instance.
With this topological view, \mnev{}~\cite{mnev1988universality} famously answered the connectivity question to the negative. In fact, \mnev{} proved a much stronger statement, nowadays known as \emph{\mnev{}'s universality theorem}: every semi-algebraic set\footnote{The semi-algebraic sets encompass a large variety of topological spaces and will be formally introduced later.} is stably equivalent\footnote{Stable equivalence is a very strong notion of equivalence of topological spaces.} to the realization space of some abstract order type. In particular, this implies that realization spaces can not only be disconnected, but can have arbitrarily complex topological features. Informally, we say that the order type realizability problem exhibits \emph{topological universality}.

Results like \mnev{}'s universality theorem have since been obtained for many types of problems. Richter-Gebert and Ziegler~\cite{richter1995realization} showed a similar universality theorem for the realization spaces of $4$\nobreakdash-dimensional polytopes. Datta~\cite{Datta2003nash} showed a universality theorem for the set of totally mixed Nash equilibria of three-player games. Shitov~\cite{shitov2016universality} proved universality for the spaces of factorizations of non-negative matrices into sums of non-negative rank-one matrices. Dobbins et al.\ proved universality for the spaces of nested polytopes~\cite{NestedPolytopesER}. 

The result of \mnev{} also implies that the problem of order type realizability is complete for the complexity class \ER, which recently is seeing much attention, e.g.\ \cite{artGalleryER,bertschinger2022training,miltzowCCSP,Schaefer2010_GeometryTopology}. It is known that $\NP\subseteq \ER\subseteq \mathsf{PSPACE}$, with both inclusions conjectured to be strict~\cite{Canny1988_PSPACE}. \ER{} is defined as the class of decision problems polynomial-time reducible to its canonical problem, the existential theory of the reals ($\mathsf{ETR}$). In the $\mathsf{ETR}$ problem, one is tasked to decide whether there exists a vector of real numbers that satisfies a given quantifier-free formula containing polynomial inequalities and equalities. By definition, $\mathsf{ETR}$ is both \ER-complete, and also exhibits topological universality. Reductions used to prove problems \ER-hard can often also be adapted to yield topological (or algebraic) universality results about the solution spaces of these problems, however no meta-theorem about the relationship between \ER-hardness or \ER-completeness and topological universality is known.

So far, the topological lens has only been used when considering the solution spaces of problems for which the solution space is naturally continuous, since solutions are described by real numbers. While this is the case in algebraic or geometric problems, combinatorial problems often have discrete solution spaces. However, topology can be useful to describe properties of discrete sets too, one just needs to find a natural definition of a continuous topological space. For this, the viewpoint of \emph{reconfiguration} is a natural approach.

Reconfiguration problems are widely studied in computer science and discrete mathematics~\cite{reconfigurationSurvey}. In reconfiguration problems, we are concerned with a set of configurations and the ways we can obtain one from the other, defining a \emph{flip graph}. 
Many questions can be asked about this graph, for example, whether it is connected, what its diameter is, or whether it is Hamiltonian.
We are particularly interested in reconfiguration on the set of solutions for an instance of some problem. For example, consider the set of satisfying variable assignments for a CNF formula~$\Phi$, and consider flipping of a single variable between $0$ and $1$ as the set of allowed moves. Instead of only looking at the flip graph (which topologically speaking is a $1$-dimensional topological space), we can consider interactions between multiple variables too; if all the variable assignments $x_{00}:=(x_{1,\ldots,d-2},0,0),x_{01}:=(x_{1,\ldots,d-2},0,1),x_{10}:=(x_{1,\ldots,d-2},1,0),x_{11}:=(x_{1,\ldots,d-2},1,1)$ are valid solutions, we not only consider the edges $\{x_{00},x_{01}\},\{x_{01},x_{11}\},\{x_{11},x_{10}\}$ and $\{x_{10},x_{00}\}$, but we can also ``fill in'' the \emph{2-dimensional face} between all of them. We can also do this with higher-dimensional faces, yielding a \emph{cubical complex}. In this very natural way, we get a topological space from a discrete set of solutions.
This very simple to define space can have surprisingly intricate structure, see \Cref{fig:S1} for an example where it is homeomorphic to a circle.

\begin{figure}
    \centering
    \noindent
    \begin{minipage}{0.5\textwidth}
        \centering
        \includegraphics{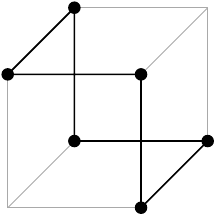}
    \end{minipage}%
    \begin{minipage}{0.5\textwidth}
    \centering
    \begin{align*}
        \phi&=(x_1\vee x_2\vee x_3)\\
        &\wedge (\neg x_1\vee \neg x_2\vee \neg x_3)\\
    \end{align*}
    \end{minipage}
    \caption{A CNF formula with a solution space homeomorphic to the circle $S^1$.}
    \label{fig:S1}
\end{figure}

In this paper, we study topological properties of such solution spaces of
\emph{boolean constraint satisfaction problems (CSPs)}.
A boolean CSP is specified by its possible constraint types, a set $S$ of logical relations. The problem instances of a boolean CSP are the formulae obtained by taking the conjunction (``and'') of finitely many constraints. The decision problem $SAT(S)$ is to decide satisfiability of a given such formula.  

Boolean CSPs are a very versatile and interesting class of problems. First off, they encompass a wide range of problems, in particular all common SAT variants. Furthermore, their simple and uniform definition makes it possible to prove very general statements. Thanks to this, their computational complexity is very well-understood nowadays. The seminal result on the computational complexity of boolean CSPs is the following \emph{dichotomy theorem} due to Schaefer:
\begin{restatable}{theorem}{schaefer}{(Schaefer's dichotomy theorem~\cite{schaefer})}\label{thm:schaefer}
    For every finite set $S$ of logical relations, the boolean constraint satisfaction problem $SAT(S)$ is polynomial-time solvable if there exists one condition in the following list that holds for all logical relations $R\in S$:
    \begin{enumerate}
        \item $R$ is $0$-valid: $(0,\ldots,0)\in R$.
        \item $R$ is $1$-valid: $(1,\ldots,1)\in R$.
        \item Horn SAT: $R$ is equivalent to a CNF formula in which every clause has at most one non-negated variable.
        \item Dual-Horn SAT: $R$ is equivalent to a CNF formula in which every clause has at most one negated variable.
        \item 2-SAT: $R$ is equivalent to a CNF formula in which every clause has at most two literals.
        \item $R$ is affine: $R$ can be expressed as the set of solutions to a system of affine equations over $GF(2)$.
    \end{enumerate}
    If there is no such condition, $SAT(S)$ is $\mathsf{NP}$-complete.
    For $SAT_C(S)$, the boolean CSP \emph{with constants}, the same classification holds, except that only the conditions 3.--6. are used.
\end{restatable}

\subsection{Preliminaries and Results}

In this paper, we prove a topological version of Schaefer's dichotomoy theorem. Before stating our results, let us formally define the relevant notions, starting with CSPs both with and without constants.

A \emph{logical relation $R$} is a subset of $\{0,1\}^k$ for some finite $k\geq 1$, where $k$ is called its \emph{rank} $rk(R)$. For a logical relation $R$ of rank $k$, we write $R(x_1,\ldots,x_k)$ for its evaluation at the value of the given variables $x_1,\ldots,x_k$. Note that $x_1,\ldots,x_k$ do not necessarily have to be distinct.

\begin{definition}[Boolean CSP]
Given a set $S=\{R_1,\ldots\}$ of logical relations, an instance of the \emph{boolean constraint satisfaction problem} $SAT(S)$ is given by a sentence $\Phi$ that is a conjunction of relations in $S$, i.e.,
\[\Phi = \exists x_1,\ldots,x_d\in\{0,1\}^d: R_{j_1}(x_{1,1},\ldots,x_{1,rk(R_{j_1})})\wedge\ldots\wedge R_{j_n}(x_{n,1},\ldots,x_{n,rk(R_{j_n})}),\]
where for all $i,j$, $x_{i,j}\in \{x_1,\ldots,x_d\}$. We call $d\geq 1$ the \emph{dimension} of $\Phi$, and $n\geq 1$ its \emph{number of constraints}.
\end{definition}

\begin{definition}[Boolean CSP with constants]
In the above definition, if $x_{i,j}$ can additionally be the constants $0$ or $1$, we call the resulting problem the boolean CSP \emph{with constants} $SAT_C(S)$.
\end{definition}

To consider the solution space of discrete problems, we need two main ingredients: (i), a way to assign an instance of the problem its discrete set of solutions, and (ii), a way to turn this discrete set into a topological space. For boolean CSPs, we define these two ingredients as follows:

\begin{definition}[Discrete solution space]
    Given a sentence $\Phi$ of dimension $d$, its \emph{discrete solution space} $\disc{\Phi}$ is the set of vertices of the hypercube $[0,1]^d$ which correspond to a solution of $\Phi$.
\end{definition}

\begin{definition}[Induced cubical complex]
    Given a set $X$ of vertices of the hypercube $[0,1]^d$, its \emph{induced cubical complex} $\induc{X}$ is the collection of all faces $f$ of $[0,1]^d$ for which every vertex $v\in f$ is contained in $X$.
\end{definition}

For the notion of universality, we want to be able to build (up to some notions of equivalence formally defined later) any semi-algebraic set.

\begin{definition}[Semi-algebraic set~\cite{realAlgebraicGeom}]
    A semi-algebraic set is the union of the solution sets to finitely many finite systems of polynomial equalities and inequalities. More formally, a semi-algebraic set is a subset of $\R^d$ of the form
    \[\bigcup_{i=1}^s\bigcap_{j=1}^{r_i}\{x\in\R^d \;\vert\; f_{i,j}*_{i,j} 0\},\]
    where the functions $f_{i,j}$ are polynomials over $\R$ in the variables $x_1,\ldots,x_d$ and $*_{i,j}\in\{<,=\}$.
\end{definition}

Throughout the paper, we will use several standard concepts of algebraic topology, such as homeomorphism, homotopy equivalence, deformation retract and homology. We refer the reader to any textbook on algebraic topology for a thorough introduction, e.g. \cite{Bredon, Hatcher}.

\begin{definition}[Homotopy-universality]
    A problem $P$ with some function $V$ assigning each instance $\Phi\in P$ a topological space to the set of its solutions is called \emph{homotopy-universal}, if for every closed and bounded semi-algebraic set $S$, there exists an instance $\Phi\in P$, such that $V(\Phi)$ is homotopy equivalent to $S$.
\end{definition}

For boolean CSPs, the function $V$ is given by the induced cubical complex of the discrete solution space, that is, $\induc{\disc{\Phi}}$.

Homotopy-universality is a comparatively weak notion of universality. However, it is not a very robust definition. For example, consider a problem whose solution spaces consist of two disjoint identical copies of semi-algebraic sets. Such a problem should be considered topologically universal, however homotopy-universality fails to capture this. In spirit of the proof strategy in \cite{bertschinger2022topological} we suggest adding a notion of projection: When a problem has a natural notion for expressing the solution space as a sub-space of some product space, i.e., $V(P)\subseteq X\times Y$, we allow to project away one of the involved spaces $Y$. The projected solution space is then the set of points $x\in X$ for which there exists a point $y\in Y$ such that $(x,y)\in V(P)$. In the case of solutions being vertices of the hypercube, which is the product space of $d$ intervals, such a projection has a natural definition:
\begin{definition}[Projection]
    Given a subset $D\subseteq [d]$ of dimensions, the \emph{projection} $\pi_D$ maps each vertex $x\in\{0,1\}^d$ to the vertex $\hat{x}\in\{0,1\}^{d-|D|}$ by removing all entries corresponding to dimensions in~$D$. Similarly, for a subset $X\subseteq \{0,1\}^d$, we define $\pi_D(X):=\{\hat{x}\;\vert\;x\in X\}$.
\end{definition}

In the discrete setting, we first project the discrete solution space, and only turn it into a topological space afterwards.

\begin{definition}[Projection-universality]
    A problem $P$ with some function $V$ assigning each instance $\Phi\in P$ and choice of projection $\pi$ a topological space is called \emph{projection-universal}, if for every closed and bounded semi-algebraic set $S$, there exists an instance $\Phi\in P$ and a projection $\pi$, such that $V(\Phi,\pi)$ is homeomorphic to $S$.
\end{definition}
In our setting, this translates to the following definition:

\begin{definition}[Projection-universality for boolean CSPs]
    A boolean CSP $SAT(S)$ or $SAT_C(S)$ is called \emph{projection-universal}, if for every closed and bounded semi-algebraic set $S$, there exists an instance $\Phi$ of the CSP and a set of dimensions~$D$, such that $\induc{\pi_D(\disc{\Phi})}$ is homeomorphic to $S$.
\end{definition}

We note here that universality up to stable equivalence implies projection-universality. Furthermore, existing proofs of homotopy-universality such as the one in \cite{bertschinger2022topological} also imply projection-universality. In general, however, homotopy-universality and projection-universality are incomparable definitions, with neither implying the other.

We are now ready to state our main results.

\begin{theorem}\label{thm:main}
    For every finite set $S$ of logical relations, each of $SAT_C(S)$ and $SAT(S)$ is projection-universal if and only if it is categorized as $\mathsf{NP}$-complete by Schaefer's dichotomy theorem (\Cref{thm:schaefer}).
\end{theorem}
We thus get the following corollary:
\begin{corollary}
    Assuming $\P\neq\NP$, a boolean CSP (with or without constants) is projection-universal if and only if it is \NP-complete.
\end{corollary}

We furthermore show that for the boolean CSPs categorized as polynomial-time solvable by Schaefer, the problem is either trivially satisfiable (in the case of $0$-valid and $1$-valid constraints), or the solution space (and any possible projections of it) has trivial homology in all dimensions $p\geq 1$.

Finally, we show a similar dichotomy theorem for homotopy-universality in CSPs given by CNF formulae with restrictions on the numbers of positive, negative, and overall variables per clause.
\begin{restatable}{theorem}{CSPdichotomy}\label{thm:CNFdichotomy}
For $k,p,n$ with $p,n\leq k\leq n+p$, let \kpnsat{k}{p}{n} be the CSP defined by the constraints expressing all possible disjunctions of at most $k$ literals, which contain at most $p$ positive literals and at most $n$ negative literals. Then, \kpnsat{k}{p}{n} is homotopy-universal if and only if it is classified as \NP-complete by \Cref{thm:schaefer}.
\end{restatable}
However, we note that there are \NP-complete CSPs which are not homotopy-universal, such as \oneinthree.

\subsection{Discussion}

\paragraph{A hierarchy of universalities}
Various notions of topological universality have been studied in the literature. The strongest notion, homeomorphism-universality has been shown for the art gallery problem~\cite{stade2022topological} and implicitly for SAT~\cite{bertschinger2022topological}, see also \Cref{lem:SAThomeomorphism}. Universality up to stable equivalence is the oldest type of universality, since it is the one used in \mnev{}'s universality theorem~\cite{mnev1988universality}. Universality up to stable equivalence implies both homotopy-universality and projection-universality, but not vice-versa. 

While some link between topological universality and different levels of computational complexity has been suspected by several researchers, there is no consensus on the correct notion of universality. One suspected link was between \ER-hardness and homotopy-universality, however, recent results in \cite{bertschinger2022topological} as well as the present paper show that homotopy-universality can already be attained by problems in \NP, even CSPs with finitely many constraint types.
\ER-hardness might however still be related to universality up to stable equivalence, at least for CSPs with finitely many constraint types\footnote{Note that this excludes SAT, which is in \NP and homeomorphism-universal.}.

\paragraph{\P vs. \NP}
Our results show a structural difference between those boolean CSPs that lie in \P and those that are \NP-complete.
Such structural differences could be vital in proving $\P\neq\NP$, however, since we only consider boolean CSPs we are very far from making progress on this fundamental problem. We discuss the possibilities for extending our results to larger classes of problems further below. On the other hand, combined with the results in this paper, a result showing that no problem in $P$ can have homotopy-universality or projection-universality would imply $\P\neq\NP$.

While we consider the topology of the solution spaces of problems, previous work considering solution spaces in the eye of computational complexity has focused on their geometry. In the 1980s, there have been efforts to prove $\P=\NP$ by giving an LP formulation of the traveling salesman problem (TSP)~\cite{swart1986}. However, the proposed formulation was later refuted by Yannakakis~\cite{yannakakis1991expressing}, who showed that \emph{symmetric} LPs for the TSP must have exponential size. Recently, Fiorini et al.~\cite{fiorini2012linear} strengthened this result and proved that every LP that projects to the TSP, maximum cut, or maximum stable set polytopes must have a superpolynomial number of constraints. The crucial tool in this series of works is the notion of \emph{extension complexity} of a polytope $P$, a measure that formalizes the necessary size of every alternative polytope $Q$ that projects to $P$. The extension complexities of the polytopes associated with various problems have since been studied~\cite{aprile2022,avis2013,goos2018,kaibel2015,pokutta2013,shitov2016universality}.

\paragraph{Algebraic Universality}
\emph{Algebraic universality} is another notion of universality that often shows up in the context of algebraic and geometric \ER-hard problems. A problem is algebraically universal, if for every algebraic number $x$, there exists an instance of the problem that does not have a solution over fields not including $x$, but has a solution over $\mathbb{Q}$ extended by $x$. The ETR problem exhibits algebraic universality. Furthermore, algebraic universality is often preserved by reductions, even by reductions aimed at proving \ER-hardness. Thus, one can find many instances of algebraic universality being proven as a byproduct from \ER-hardness proofs~\cite{artGalleryER,bertschinger2022training,NestedPolytopesER,TensorRankUniversal,shitovFactorization}. Note that algebraic universality cannot occur in problems with \emph{open} solution spaces, such as the \ER-hard ETR variant used in \cite{openSolutionSpace}. We are not aware of any problems contained in \NP exhibiting algebraic universality.

\paragraph{Future Work}
In this work we only considered boolean CSPs. To obtain more convincing arguments for the suspected link between projection-universality and \NP-hardness, one would need to consider more general classes of problems. Interesting results in this line of research would be statements such as ``if a problem's solution space cannot be too complicated\footnote{For example, if the solution spaces are always contractible, always have trivial homology, etc.}, then it can be solved in polynomial time'' or ``a problem is \NP-hard if and only if it is topologically universal under some fixed equivalence relation''. Note that even both of these statements together would not resolve the $\P$ vs.\ $\NP$ problem.

To generalize our results, one needs to consider a larger class of problems with a common encoding and a common solution space definition. Natural candidates are more general classes of CSPs. However, their computational complexity needs to be well-understood before a link between topological and computational complexity can be investigated. Dichotomy theorems have been found for various such classes, for example for CSPs on finite domains by two independent and simultaneous papers by Zhuk~\cite{finiteDomainCSP} and Bulatov~\cite{bulatovGeneralCSP}. Both of these approaches also work for infinite domain CSPs~\cite{bulatovGeneralCSP,finiteDomainCSP,zhuk2018modification}. Miltzow and Schmiermann showed that \emph{continuous} constraint satisfaction problems are \ER-complete as soon as an addition constraint and any curved (i.e., not piecewise linear) constraint can be encoded~\cite{miltzowCCSP}. Bodirsky and Pinsker established a classification similar to Schaefer's for propositional logic of graphs~\cite{bodirskyGraphs}. Dichotomy theorems have also been obtained for optimization~\cite{kolmogorovGeneralValued,kozikValued} and counting~\cite{creignouCounting} versions of CSPs.

Another interesting approach is to study the topology of solution spaces of other CSPs in $\P$ or of problems which are conjectured to be \NP-intermediate, such as the graph isomorphism problem. In particular, a result showing that the solution spaces are not universal but can still have arbitrarily complicated topology in some sense would be further evidence for a relationship between complexity of solution spaces and computational complexity.

Finally, it is interesting to study universality through a more fine-grained lens.
Already Schaefer~\cite{schaefer} and later also Allender et al.~\cite{allenderRefining} refined the classification of the complexity of boolean CSPs into more fine-grained classes than just \P vs \NP-complete. This refined classification could also be interesting to investigate from a topological view. For example, is there a \emph{topological complexity measure} that is more fine-grained than topological universality vs.\ non-universality that agrees with this more fine-grained classification?

\subsection{Proof Techniques and Paper Overview}

To prove universality of the \NP-complete boolean CSPs, we show projection-universality and homotopy-universality of 3-SAT. For this, we extend the techniques of Bertschinger et al.~\cite{bertschinger2022topological} used to show homotopy-universality of the art gallery problem. Analogous methods work for the other \NP-complete SAT variants in \Cref{thm:CNFdichotomy}. We then show that by applying the reduction used by Schaefer~\cite{schaefer} we get projection-universality of all boolean CSPs classified as \NP-complete. This is done in \Cref{sec:universality}.

To prove non-universality of the six types of boolean CSPs characterized as polynomial-time solvable by Schaefer, we have to consider a few cases separately. In the case of trivially satisfiable problems, we can exclude the empty set from occurring as a solution space. In the case of affine constraints, the solution space can be observed to be a disjoint union of faces. In the cases of 2-SAT, Horn SAT, and Dual-Horn SAT, we show that the \emph{union} of the solution spaces of $k$ individual constraints cannot possess non-trivial $p$-homology for $p\geq k$. We then use the Mayer-Vietoris sequence to prove that the \emph{intersection} of these solution spaces (thus the solution space of the whole formula) has trivial $p$-homology for all $p\geq 1$. This already implies that we do not have homotopy-universality. To show the absence of projection-universality, we show that any projection of a 2-SAT, Horn SAT or Dual-Horn SAT instance is again a 2-SAT, Horn SAT or Dual-Horn SAT instance, respectively. This is done in \Cref{sec:nonuniversality}.

\section{Topological Universality}\label{sec:universality}
We first show that every closed and bounded semi-algebraic set $S$ can be homeomorphically represented as the induced cubical complex of some vertices in the hypercube. Using this, we show homeomorphism-universality of SAT. Applying the classic reduction of SAT to 3-SAT we show projection-universality and homotopy-universality of 3-SAT. With this, we then show projection-universality of all CSPs classified as polynomial-time by \Cref{thm:schaefer}. We also strengthen this to homotopy-universality for CNF variants.

\begin{lemma}\label{lem:semialgebraicsetasinducedcomplex}
    For every closed and bounded semi-algebraic set $S$, there exists a dimension $d$, such that $S$ is homeomorphic to the induced cubical complex \induc{V} of some set $V\subseteq \{0,1\}^d$.
\end{lemma}
\begin{proof}
    Every semi-algebraic set $S$ can be triangulated~\cite{hironaka1975triangulationsOfAlgebraic,hofmann2009triangulationsOfSemialgebraic}. Simplicial complexes can be turned into homeomorphic cubical complexes~\cite[Lemma 1.2]{blass1972simplicialtocubical}. This can for example be achieved as follows: Given a simplicial complex $K$ on $d$ vertices, we map each vertex $i\in[d]$ to a unit vector $e_i$ in the $d$-dimensional hypercube $[0,1]^d$. For each face $f\subseteq[d]$, we add the induced cubical complex of the set
    \[V_f:=\big\{v\in\{0,1\}^d\;\vert\; v\neq 0\text{ and }v\leq \sum_{i\in f}e_i\big\}.\]
    One can easily verify that this yields a homeomorphic cubical complex. Furthermore, one can verify that the induced cubical complex of $V:=\bigcup_{f\in K}V_f$ is equal to the union of the induced cubical complexes of each $V_f$. Thus, the cubical complex induced by $V$ is indeed homeomorphic to $K$ and to the semi-algebraic set $S$.
\end{proof}

From this we already get the following universality statement.

\begin{lemma}\label{lem:SAThomeomorphism}
    SAT is homeomorphism-universal.
\end{lemma}

\begin{proof}
    Let $S$ be a closed and bounded semi-algebraic set. By \Cref{lem:semialgebraicsetasinducedcomplex} there exists an integer $d$ such that $S$ is homeomorphic to the induced  cubical complex $\induc{V}$ of some set $V\subseteq\{0,1\}^d$.
    Then there exists a CNF formula $\Psi$ with $\disc{\Psi}=V$. Such a formula can for example be obtained by building a DNF formula and converting it to CNF.
\end{proof}

Applying the classic reduction from SAT to 3-SAT, we get projection-universality of 3-SAT.
\begin{lemma}\label{lem:3SATprojectionuniversal}
For $S$ being the set of all logical relations of rank $k\leq 3$ that can be expressed as the disjunction of $k$ (possibly negated) variables, $SAT(S)$ (also called 3-SAT) is projection-universal.    
\end{lemma}
\begin{proof}
    Let $S$ be a closed and bounded semi-algebraic set. By \Cref{lem:SAThomeomorphism} there exists a CNF formula~$\Psi$ whose cubical complex induced by its solutions is homeomorphic to $S$.
    We now translate~$\Psi$ into a 3-SAT formula $\Phi$.
    For every clause $C$ of $\Psi$ with at least four literals, we apply the classical translation of CNF into 3-CNF. For simplicity, assume that there is only one such clause; the following arguments can be applied to each clause independently. Let $C=z_1\vee z_2\vee\ldots\vee z_k$ with $k\geq 4$ and $z_i\in\{x_1,\neg x_1,\ldots, x_d,\neg x_d\}$. $C$ is replaced by the formula
    \[F = (z_1\vee z_2 \vee y_{1}) \wedge (\neg y_{1}\vee z_3\vee y_{2})\wedge \ldots\wedge (\neg y_{k-4}\vee z_{k-2}\vee y_{k-3})\wedge (\neg y_{k-3}\vee z_{k-1}\vee z_{k}).\]
    It is well-known that for this construction, $x\in\{0,1\}^d$ is in $\disc{\Psi}$ if and only if there exists $y\in\{0,1\}^{k-3}$ such that $(x,y)\in\disc{\Phi}$.
    Thus, considering $D$ to be the set of dimensions corresponding to $y_1,\ldots,y_{k-3}$, we have $\induc{\pi_D(\disc{\Phi})}=\induc{\disc{\Psi}}$, which is homeomorphic to $S$, proving projection-universality.
\end{proof}

To extend this further to homotopy-universality of 3-SAT, we follow a similar strategy as used in \cite{bertschinger2022topological}  and use the following version of a mapping theorem due to Smale~\cite{smale1957vietoris}:
\begin{lemma}[{\cite[Lemma 3.6]{bertschinger2022topological}}]\label{lem:smale}
    Let $X$ and $Y$ be connected cubical complexes. 
    Assume that there exists a continuous proper surjective map $f:X\rightarrow Y$ such that for each point
    $y\in Y$, its pre-image $f^{-1}(y)$ is contractible. Then $X$ and $Y$ are homotopy equivalent.
\end{lemma}

\begin{lemma}\label{lem:3SAThomotopyuniversal}
     3-SAT is homotopy-universal.
\end{lemma}
\begin{proof}
    We prove that the projection $\pi_D$ in the proof of \Cref{lem:3SATprojectionuniversal} in fact induces a homotopy equivalence between $\induc{\pi_D(\disc{\Phi})}$ and $\induc{\disc{\Phi}}$. Consider the map $f:[0,1]^{d}\times [0,1]^{k-3}\rightarrow [0,1]^d$ that maps each pair $(x,y)$ to $x$. This map is simply the extension of $\pi_D$ to the whole cube $[0,1]^{d+k-3}$.
    
    We restrict the domain of $f$ to $\induc{\disc{\Phi}}$, and claim that it then fulfills the conditions of \Cref{lem:smale}. Firstly, it is clear that $f$ is continuous and that it is surjective onto~$\induc{\pi_D(\disc{\Phi})}$. Secondly, for every compact set $X\subseteq \induc{\pi_D(\disc{\Phi})}$, $f^{-1}(X)$ is also compact, thus $f$ is proper.
    
    Finally, let us consider the pre-image $f^{-1}(x)$ for every $x\in \induc{\pi_D(\disc{\Phi})}$. We have that $f^{-1}(x)=\{x\}\times Y$, where $Y$ is some cubical complex. Let $V_x$ be the set of vertices of the minimal face of $\{0,1\}^d$ that contains $x$ in its interior. Then, we can see that $Y$ is induced by the vertices $y'\in\{0,1\}^{k-3}$ such that for all $x'\in V_x$, $(x',y')\in\disc{\Phi}$. We aim to prove that $Y$ is contractible by showing that it is homeomorphic to a cube.
    
    To achieve this, we partition the set of clauses of the formula $F=F_1\wedge\ldots\wedge F_{k-2}$ into $F^*$ and $F^-$, where a clause of $F_i$ is put into $F^*$ if it is already fulfilled by all $x'\in V_x$, and into $F^-$, otherwise. 
    Consider all sequences $F_i,\ldots,F_{i+\ell}$ such that $F_i,F_{i+\ell}\in F^*$ and $F_j\in F^-$ for all $j\in \{i+1,\ldots,i+\ell-1\}$. For each such sequence, $y_i,\ldots,y_{i+\ell-1}$ can be assigned values independently from all other $y_j$. Furthermore, for $y_i,\ldots,y_{i+\ell-1}$, the fulfilling assignments are the assignments $y_i=\ldots =y_{i+k}=1$ and $y_{i+k+1}=\ldots=y_{i+\ell-1}=0$. In the induced cubical complex $Y$, this is homeomorphic to a path. Since the assignment for each sequence is independent from the others, we thus have that $Y$ is homeomorphic to the Cartesian product of $|F^*|-1$ paths, i.e., a $(|F^*|-1)$-dimensional cube.
    Thus, for every $x\in\induc{\pi_D(\disc{\Phi})}$, $f^{-1}(x)$ is contractible.

    We can therefore apply \Cref{lem:smale} individually to each connected component of $\induc{\disc{\Phi}}$, and get homotopy equivalence between $\induc{\disc{\Phi}}$ and $\induc{\pi_D(\disc{\Phi})}$ as desired. Thus, the lemma follows.
\end{proof}

Following the proof of \Cref{thm:schaefer} as stated in \cite{schaefer}, one quickly obtains projection-universality of the \NP-complete $SAT(S)$ and $SAT_C(S)$ problems:
\begin{lemma}\label{lem:CSPprojectionUniversal}
    Let $S$ be a finite set of logical relations. If \Cref{thm:schaefer} classifies $SAT(S)$ ($SAT_C(S)$, respectively) as \NP-complete, then it is projection-universal.
\end{lemma}
\begin{proof}
    We provide the proof for the $SAT(S)$ case, however the proof for $SAT_C(S)$ works in exactly the same way.
    
    In the proof of \NP-hardness in \cite{schaefer}, it is proven that for every 3-SAT formula $\Phi$ on the variables $x_1,\ldots,x_d$ there exists a formula $\Psi\in SAT(S)$ on the variables $x_1,\ldots,x_d,y_1,\ldots,y_k$, such that $x_1,\ldots,x_d$ solves $\Phi$ if and only if there exist $y_1,\ldots,y_k$, such that $x_1,\ldots,x_d,y_1,\ldots,y_k$ solves $\Psi$. Thus, letting $D$ be the set of dimensions corresponding to $y_1,\ldots,y_k$, $\pi_D(\disc{\Psi})=\disc{\Phi}$. Since 3-SAT is projection-universal, and since projections can be chained ($\pi_D\circ\pi_{D'}=\pi_{D\cup D'}$), it follows that $SAT(S)$ is projection-universal.
\end{proof}
This finishes the proof of the ``if'' direction of our main result, \Cref{thm:main}. While \Cref{lem:CSPprojectionUniversal} does not hold for homotopy-universality (as we will prove in \Cref{sec:nonuniversality}), a homotopy-universality result does hold for a sub-class of CSPs. We prove the ``if'' direction of \Cref{thm:CNFdichotomy}:
\begin{lemma}\label{lem:CNFUniversal}
    For $k,p,n$ with $p,n\leq k\leq n+p$, let \kpnsat{k}{p}{n} be the CSP defined by the constraints expressing all possible disjunctions of at most $k$ literals, which contain at most $p$ positive literals and at most $n$ negative literals. Then, \kpnsat{k}{p}{n} is homotopy-universal if it is classified as \NP-complete by \Cref{thm:schaefer}.
\end{lemma}
\begin{proof}
    If $k\leq 2$, $p\leq 1$, or $n\leq 1$, \kpnsat{k}{p}{n} is expressible in 2-SAT, Horn SAT, or Dual-Horn SAT, and thus classified as polynomial-time solvable. Otherwise, if $k,p,n\geq 3$, \kpnsat{k}{p}{n} contains the homotopy-universal 3-SAT and is thus also homotopy-universal.

    Thus, the only missing cases are $k\geq 3$, $p=2$, $n\geq \max (2,k-2)$. We prove homotopy-universality for the simplest of these cases, namely \kpnsat{3}{2}{2}, which is contained by all other such cases. We prove this by translating every 3-SAT formula $\Phi$ into a \kpnsat{3}{2}{2} formula $\Psi$ with homotopy equivalent solution space. To achieve this, we simply translate a clause of the form $(x\vee y\vee z)$ into the formula $(x\vee y\vee \neg\alpha)\wedge (\alpha\vee z)$ for a new variable $\alpha$. We apply the symmetric construction to clauses of the form $(\neg x\vee \neg y\vee \neg z)$. Homotopy equivalence of $\induc{\disc{\Phi}}$ and $\induc{\disc{\Psi}}$ follows from exactly the same arguments as in the proof of \Cref{lem:3SAThomotopyuniversal}
\end{proof}

\section{Absence of Topological Universality}\label{sec:nonuniversality}
To prove the ``only if'' part of \Cref{thm:main}, we first show that for each of the six cases of \Cref{thm:schaefer}, $SAT(S)$ (and for cases 3.--6., $SAT_C(S)$) is not homotopy-universal. Since by definition $SAT_C(S)\supseteq SAT(S)$, non-universality of $SAT_C(S)$ implies non-universality of $SAT(S)$. We then extend this to projection-universality in \Cref{sec:nonuniversalityProjection}. 

\subsection{Absence of Homotopy-Universality}\label{sec:nonuniversalityHomotopy}
We begin with the easy cases 1. ($0$-valid), 2. ($1$-valid), in which case we only need to show non-universality for $SAT(S)$, as $SAT_C(S)$ is classified as \NP-complete by \Cref{thm:schaefer} for these cases.
\begin{lemma}\label{lem:nonuniversal0valid}
    For a finite set $S$ of $0$-valid logical relations, $SAT(S)$ is \emph{not} homotopy-universal.
\end{lemma}
\begin{proof}
    The empty set is a semi-algebraic set. However, for any instance $\Phi$ of $SAT(S)$, $(0,\ldots,0)\in \disc{\Phi}$. Since a non-empty set cannot be homotopy equivalent to the empty set, $SAT(S)$ is not homotopy-universal.
\end{proof}
Analogously, we get the same result for $1$-valid logical relations.
\begin{lemma}\label{lem:nonuniversal1valid}
    For a finite set $S$ of $1$-valid logical relations, $SAT(S)$ is \emph{not} homotopy-universal.
\end{lemma}

The remaining cases are classified as \NP-complete for both $SAT_C(S)$ and $SAT(S)$. Recall that non-universality of $SAT_C(S)$ implies non-universality of $SAT(S)$, we thus only show the former. First, we handle case 6. (affine).
\begin{lemma}\label{lem:nonuniversalLinear}
     For a finite set $S$ of affine logical relations, $SAT_C(S)$ is \emph{not} homotopy-universal.
\end{lemma}
\begin{proof}
    We can represent every formula in $SAT_C(S)$ by a formula $\Phi$ where every constraint of $\Phi$ is of the form $x_1\oplus\ldots\oplus x_k=c$ for $c\in\{0,1\}$ and $x_1,\ldots,x_k$ are all distinct. Let $n$ be the dimension of $\Phi$. Let us consider $\induc{\disc{\Phi}}$. Without loss of generality, every variable $x_i$ for $i\in [n]$ occurs in some constraint of $\Phi$: otherwise, the projection in these dimensions is a deformation retract, which preserves homotopy equivalence. If $x\in\{0,1\}^n$ is a solution to $\Phi$, we know that no neighboring assignment $x'$ (with $x'_j=x_j$ for all $j\neq i$ for some $i$) can be a solution, since every constraint in which $x_i$ occurs must be violated by $x'$. Thus, $\induc{\disc{\Phi}}$ contains no edges and thus also not any faces of larger dimension. We conclude that for every $\Phi$, $\induc{\disc{\Phi}}$ is homotopy equivalent to a discrete topological space.
\end{proof}

For the remaining three cases, 3. (Horn SAT), 4. (Dual-Horn SAT), and 5. (2-SAT), we aim to prove that the solution space $\induc{\disc{\Phi}}$ of every instance $\Phi$ has trivial homology groups for all dimensions $p\geq 1$. Without loss of generality, we assume that each constraint in $\Phi$ is actually a single clause (disjunction) of the respective SAT variant. Note that disjunctions containing a constant~$1$ are trivially satisfied and can be removed. Similarly, disjunctions containing a constant~$0$ are equivalent to the same disjunction with the constant removed. We thus only need to consider formulae without constants.

To prove the absence of homotopy-universality, we run one single proof by induction, where only the base case has to be proven separately for each of the three variants. The proof makes heavy use of a theorem due to Mayer and Vietoris~\cite{Vietoris1930}, with a more modern version due to Eilenberg and Steenrod~\cite{MayerVietorisModern} known as the \emph{Mayer-Vietoris sequence}. We first need to define exact sequences:

\begin{definition}[Exact sequence]
An \emph{exact sequence} $G_0\overset{f_1}{\rightarrow}G_1\overset{f_2}{\rightarrow}\ldots\overset{f_n}{\rightarrow}G_n$ is a sequence of groups $G_0,\ldots,G_n$ and homomorphisms $f_1,\ldots,f_n$ with $f_i:G_{i-1}\rightarrow G_i$, such that for all $1\leq i\leq n-1$, $\im{f_i}=\ker{f_{i+1}}$.
\end{definition}

We are now ready to state our main tool:

\begin{theorem}[Mayer-Vietoris sequence]\label{thm:mayervietoris}
    Let $X$ be a topological space, and $A,B\subseteq X$ be two subspaces whose interiors cover~$X$ (the interiors do not need to be disjoint). Then, there is an exact sequence
\[ \ldots\rightarrow H_{p+1}(A)\oplus H_{p+1}(B) \rightarrow H_{p+1}(A\cup B) \rightarrow H_p(A\cap B)\rightarrow H_p(A)\oplus H_p(B) \rightarrow \ldots \rightarrow H_0(A\cup B) .\]
\end{theorem}

Here, $G\oplus H$ denotes the direct sum of groups. Note that $0\oplus 0\isomorph 0$.
We use the following simple and well-known fact:
\begin{fact}\label{lem:isomorphInSequence}
    Let $0\overset{e}{\rightarrow}A\overset{f}{\rightarrow}B\overset{g}{\rightarrow}0$ be an exact sequence. Then, $A\isomorph B$.
\end{fact}

Given a sentence $\Phi=\exists x_1,\ldots,x_d:C_1\wedge\ldots\wedge C_n$, we consider the solution spaces of each of the clauses $C_i$ individually. The \emph{wedge} $W_i$ is the solution space $\induc{\disc{\Phi_i}}$ of $\Phi_i:=\exists x_1,\ldots,x_d:C_i$. Note that $\induc{\disc{\Phi}}=\bigcap_{i=1}^n W_i$.
We first claim that in the three considered settings, \emph{unions} of wedges have very limited limited topological complexity. This is the only part of the proof that needs to be shown for each setting separately. We defer the proof of this lemma to the end of the section.
\begin{lemma}\label{lem:simpleunions}
    Given $n$ wedges $W_1,\ldots,W_{n}$ of some 2-SAT, Horn SAT, or Dual-Horn SAT formula $\Phi$ with $n'\geq n$ clauses, we have that
    \[H_p\big(\bigcup_{i=1}^{n}W_i\big)\isomorph 0, \text{ for all }p\geq n.\]
\end{lemma}

For the rest of the proof, we assume that $\Phi$ is either a 2-SAT, Horn SAT, or Dual-Horn SAT formula, and do not need to distinguish between these settings.

\begin{claim}\label{claim:mayervietorisApplicable}
    \Cref{thm:mayervietoris} also applies to any $A$ and $B$ that are composed of unions and intersections of wedges.
\end{claim} 
\begin{proof}
    Such spaces $A$ and $B$ are unions of faces of the $d$-dimensional hypercube. There exist open supersets $A'$ and $B'$ obtained by slightly thickening $A$ and $B$, such that $A$, $B$, $A\cap B$, and $A\cup B$ are deformation retracts of $A'$, $B'$, $A'\cap B'$, and $A'\cup B'$, respectively. Since $A'$ and $B'$ are open, their interiors cover their union $A'\cup B'$, and \Cref{thm:mayervietoris}. However, the thickening has not changed any of the homology groups appearing in the sequence of \Cref{thm:mayervietoris}, and the sequence is thus also exact for $A$ and $B$.
\end{proof}

\begin{lemma}\label{lem:cupcap}
    For every $n\geq 1$, $p\geq 1$, and $0\leq k\leq n-1$,
\[H_{p+k}\big(\bigcap_{i=1}^{n-k} W_i \cup \bigcup_{j=n-k+1}^{n} W_j\big)\isomorph 0.\]
\end{lemma}
\begin{proof}
We prove this by total induction on $n$. The base case $n=1$ trivially follows from \Cref{lem:simpleunions}.
For the induction step we assume that the theorem holds for all $n'<n$, and prove it for $n$. Using this induction hypothesis, we prove the following for all $k$ with $0\leq k\leq n-2$:
\begin{equation}
    H_{p+k+1}\Big((\bigcap_{i=1}^{n-k-1}W_i)\cup(W_{n-k}\cup \bigcup_{j=n-k+1}^n W_j)\Big)\isomorph H_{p+k}\Big((\bigcap_{i=1}^{n-k-1}W_i\cap W_{n-k})\cup(\bigcup_{j=n-k+1}^n W_j)\Big).\label{eqn:hint}
\end{equation}
Combined with \Cref{lem:simpleunions}, \Cref{eqn:hint} directly implies the theorem. It thus only remains to prove Equation \ref{eqn:hint}. We define the three sets 
\[A:=\bigcap_{i=1}^{n-k-1}W_i, \;\; W:=W_{n-k}, \;\; B:=\bigcup_{j=n-k+1}^n W_j.\] Now note that \Cref{eqn:hint} is equivalent to $H_{p+k+1}(A\cup(W\cup B))\isomorph H_{p+k}((A\cap W)\cup B)$. Let $X=A\cup B$ and $Y=W\cup B$. Then, 
\[X\cup Y=A\cup(W\cup B)\;\;\text{and}\;\;X\cap Y = (A\cup B)\cap (W\cup B)=(A\cap W)\cup B.\]
We thus only need to prove that $H_{p+k+1}(X\cup Y)\isomorph H_{p+k}(X\cap Y)$. By \Cref{thm:mayervietoris} and \Cref{claim:mayervietorisApplicable}, there exists an exact sequence
\begin{equation}
    H_{p+k+1}(X)\oplus H_{p+k+1}(Y)\rightarrow H_{p+k+1}(X\cup Y) \rightarrow H_{p+k}(X\cap Y)\rightarrow H_{p+k}(X)\oplus H_{p+k}(Y).\label{eqn:sequence}
\end{equation}
We can now use the induction hypothesis to show that the left-most and right-most group in this sequence are $0$, and thus the two in the middle are isomorphic as desired. To achieve this, we see that $X$ can be written as 
\[X=\bigcap_{i=1}^{n-k-1}W_i \cup \bigcup_{j=n-k+1}^n W_j=\bigcap_{i=1}^{(n-1)-k}W_i\cup \bigcup_{j=(n-1)-k+1}^{n-1} W_{j+1},\]
which is of the correct form (up to re-indexing) to apply the induction hypothesis with $n'=n-1$, yielding $H_{p+k}(X)\isomorph 0$. Since $p$ can always be increased, also $H_{p+k+1}(X)\isomorph 0$.

We apply a similar approach to $Y$: $Y$ can be written as 
\[Y=\bigcap_{i=n-k}^{n-k}W_i \cup \bigcup_{j=n-k+1}^n W_j=\bigcap_{i=1}^{(n-(n-k-1))-k}W_{i+(n-k-1)} \cup \bigcup_{j=(n-(n-k-1))-k+1}^{n-(n-k-1)} W_{j+(n-k-1)},\]
which is of the correct form to apply the induction hypothesis with $n'=n-(n-k-1)\leq n-1$, to also yield $H_{p+k}(Y)\isomorph 0$ and $H_{p+k+1}(Y)\isomorph 0$.

Thus, $H_{p+k+1}(X)\oplus H_{p+k+1}(Y)\isomorph 0\oplus 0\isomorph 0$, as well as $H_{p+k}(X)\oplus H_{p+k}(Y)\isomorph 0\oplus 0\isomorph 0$, and thus by exactness of the sequence in \Cref{eqn:sequence}, $H_{p+k+1}(X\cup Y)\isomorph H_{p+k}(X\cap Y)$. We have thus proven \Cref{eqn:hint}, and the theorem follows.
\end{proof}

As a consequence of \Cref{lem:cupcap} for $k=0$, we get the following desired corollaries:
\begin{corollary}\label{cor:trivialHomology2Horn}
    For every 2-SAT, Horn SAT, or dual-Horn SAT formula $\Phi$, $H_p(\induc{\disc{\Phi}})\isomorph 0$ for all $p\geq 1$.
\end{corollary}
\begin{corollary}\label{cor:nonuniversalSAT}
    For a finite set $S$ of logical relations that are all equivalent to either 2-SAT, Horn SAT, or Dual-Horn SAT, $SAT_C(S)$ is \emph{not} homotopy-universal.
\end{corollary}

Furthermore, together with \Cref{lem:CNFUniversal}, this finishes the proof of the homotopy-universality dichotomy for CNF variants, \Cref{thm:CNFdichotomy}.

\subsubsection{Proof of Lemma \ref{lem:simpleunions}}
We prove this for 2-SAT and Horn SAT separately, and the claim for Dual-Horn SAT follows by symmetry.

\begin{proof}[Proof of \Cref{lem:simpleunions} for 2-SAT]
    Without loss of generality all variables $x_1,\ldots,x_d$ occur in at least one clause describing some wedge $W_j$, since otherwise there is a deformation retract of $\bigcup_{i=1}^n W_i$ to a subcube only spanned by the variables occurring at least once. Furthermore, we can assume that none of the wedges is described by a clause of the form $(x_i\vee \neg x_i)$ since otherwise the considered union is equal to the whole hypercube, which is contractible.
    
    If any variable $x_j$ only occurs in a single clause describing one of the wedges, without loss of generality positively, there is a deformation retract of $\bigcup_{i=1}^n W_i$ onto the facet $\{x\;|\; x_j=1\}\subseteq \bigcup_{i=1}^n W_i$, since for any vertex $v\in \bigcup_{i=1}^n W_i$, the vertex $v'$ with $v'_k=v_k$ for $j\neq j$ and $v'_j=1$ is also a solution. Since the facet $\{x\;|\; x_j=1\}$ is contractible, it has trivial homology groups $H_p$ for all $p\geq 1$, and thus so does $\bigcup_{i=1}^n W_i$.
    
    If no variable occurs only in a single clause, every variable must occur at least twice. Since the clause describing each wedge contains only two variables, we thus must have $d\leq n$. Since every wedge is a $(d-1)$-dimensional object, $\bigcup_{i=1}^n W_i$ is at most $(n-1)$-dimensional, and thus has trivial $p$-homology for all $p\geq n$.
\end{proof}

The proof works very similarly for Horn SAT:

\begin{proof}[Proof of \Cref{lem:simpleunions} for Horn SAT]
    As in the proof for 2-SAT, without loss of generality all variables $x_1,\ldots,x_d$ occur in at least one clause describing one of the wedges $W_j$, and no clause contains both the literals $x_i$ and $\neg x_i$.

    If any variable $x_j$ only occurs negatively, there is again a deformation retract of $\bigcup_{i=1}^n W_i$ onto the facet $\{x\;|\; x_j=0\}\subseteq \bigcup_{i=1}^n W_i$.

    If every variable occurs positively in at least one clause, we must have at least $d$ clauses, i.e., $d\leq n$. Since again every wedge is a $(d-1)$-dimensional object, $\bigcup_{i=1}^n W_i$ has trivial $p$-homology for all $p\geq n\geq d$.
\end{proof}

\subsection{Absence of Projection-Universality}\label{sec:nonuniversalityProjection}
We now make the necessary adaptations for the results of \Cref{sec:nonuniversalityHomotopy} to translate to projection-universality.

The proofs of \Cref{lem:nonuniversal0valid,lem:nonuniversal1valid} translate directly, since projecting the non-empty set $\disc{\Phi}$ always yields a non-empty set, and thus $\induc{\pi(\disc{\Phi})}$ is non-empty too.

For all the other cases $SAT_C(S)$, we will show that for every instance $\Phi$ of $SAT_C(S)$ and every projection $\pi_D$, there exists an instance $\Phi'$ of $SAT_C(S')$, such that $S'$ fulfills the same condition of Schaefer's dichotomy theorem as $S$, and $\disc{\Phi'}=\pi_D(\disc{\Phi})$.
Recall that by \Cref{cor:trivialHomology2Horn} and the proof of \Cref{lem:nonuniversalLinear}, there exists a semi-algebraic set $X$, such that for no instance $\Phi'$ of $SAT_C(S')$,  $\induc{\disc{\Phi'}}$ is homotopy equivalent to $X$. Since two spaces being homeomorphic implies that they are homotopy equivalent, $\induc{\pi_D(\disc{\Phi})}$ can thus also not be homeomorphic to $X$. This shows that $SAT_C(S)$ is not projection-universal.

\begin{lemma}\label{lem:projectionAffine}
    Let $S$ be a finite set of affine logical relations. Then, for every $d$-dimensional instance $\Phi\in SAT_C(S)$ and every $D\subseteq [d]$, there exists an instance $\Phi'$ of $SAT_C(S')$ for some finite set $S'$ of affine logical relations, such that $\disc{\Phi'}=\pi_D(\disc{\Phi})$.
\end{lemma}
\begin{proof}
    By \cite[Lemma 3.1A]{schaefer}, a relation $R$ is affine if and only if for every three $s_1,s_2,s_3\in R$, we also have $s_1\oplus s_2\oplus s_3\in R$.

    Let $[d]$ be the set of variables of $\Phi$. Without loss of generality, let $D=\{x_{d-|D|+1},\ldots,x_d\}$. A vertex $x\in\{0,1\}^{d-|D|}$ is in $\pi_D(\disc{\Phi})$ if and only if there exists $y\in\{0,1\}^{|D|}$ such that $(x,y)\in\disc{\Phi}$.

    We build a single relation $R$ of rank $d-|D|$ consisting of exactly $\pi_D(\disc{\Phi})$ and show that $R$ is affine:
    For every three $x_1,x_2,x_3\in R=\pi_D(\disc{\Phi})$, we have
    \[\exists y_1,y_2,y_3: (x_1,y_1), (x_2,y_2), (x_3,y_3)\in \disc{\Phi}.\]
    However, since $\disc{\Phi}$ is the set of solutions to a system of affine equations over $GF(2)$, we must have that $(x_1\oplus x_2\oplus x_3, y_1\oplus y_2\oplus y_3)\in\disc{\Phi}$, and thus $(x_1\oplus x_2\oplus x_3)\in \pi_D(\disc{\Phi})=R$. Thus, $R$ is affine.

    We conclude that the formula $\exists x_1,\ldots,x_{d-|D|}:R(x_1,\ldots,x_{d-|D|})$, which is an instance of $SAT_C(S')$ for $S'=\{R\}$, fulfills the necessary conditions.
\end{proof}

\begin{lemma}\label{lem:projection2SAT}
    Let $S$ be a finite set of logical relations each equivalent to a 2-SAT formula. Then, for every $d$-dimensional instance $\Phi\in SAT_C(S)$ and every $D\subseteq [d]$, there exists a 2-SAT formula $\Phi'$ such that $\disc{\Phi'}=\pi_D(\disc{\Phi})$.
\end{lemma}
\begin{proof}
    We only consider projection by a single variable, i.e., $D=\{x_d\}$. The statement then follows by induction, since $\pi_{D'}\circ\pi_D=\pi_{D'\cup D}$. Without loss of generality, let $\Phi$ be given as a 2-SAT formula.

    We split the set of clauses of $\Phi$ into two subformulas: The clauses of $\Phi$ containing the variable $x_d$ are put into $\Phi^\in$, and those not containing $x_d$ into $\Phi^{\not\in}$. Furthermore, $\Phi^\in$ is further split into two subformulas: clauses containing the positive literal $x_d$ are put into $\Phi^{\in+}$, and those containing the negative literal $\neg x_d$ into $\Phi^{\in-}$. We can assume there exists no trivial clause $(x_d\vee\neg x_d)$, thus $\Phi^{\in+},\Phi^{\in-},\Phi^{\not\in}$ form a proper partition of the clauses of $\Phi$.
    
    Let us now consider the projection $\pi_D(\disc{\Phi})$. We first see that 
    \[(x_{1,\ldots,d-1},0)\in \disc{\Phi^{\not\in}}\Leftrightarrow (x_{1,\ldots,d-1},1)\in \disc{\Phi^{\not\in}}.\]
    Furthermore,
    \[(x_{1,\ldots,d-1},1)\in \disc{\Phi^{\in+}}\;\text{ and }\;(x_{1,\ldots,d-1},0)\in \disc{\Phi^{\in-}}.\]
    We thus have that $(x_{1,\ldots,d-1})\in \pi_D(\disc{\Phi})$ if and only if 
    \[\bigg((x_{1,\ldots,d-1},0)\in\disc{\Phi^{\not\in}}\text{ and } \Big((x_{1,\ldots,d-1},0)\in\disc{\Phi^{\in+}} \text{ or } (x_{1,\ldots,d-1},1)\in\disc{\Phi^{\in-}}\Big)\bigg)\] 
    We can thus write a formula $\Phi':=\Psi^{\not\in}\wedge (\Psi^{\in+}\vee \Psi^{\in-})$,
    where $\Psi^{\not\in}$ is equal to $\Phi^{\not\in}$ except that it disregards $x_d$, and $\Psi^{\in+}$ and $\Psi^{\in-}$ are versions of $\Phi^{\in+}$ and $\Phi^{\in-}$ with all literals $x_d$ and $\neg x_d$ removed. Since every clause of $\Phi^{\in+}$ and $\Phi^{\in-}$ contained $x_d$, $\Psi^{\in+}$ and $\phi^{\in-}$ are now in 1-CNF. Phrased more intuitively, they are each a single conjunction of some literals. Thus, we can also write $\Phi'$ as
    \[\Psi^{\not\in}\wedge \big((\bigwedge_{\alpha\in \Psi^{\in+}}\alpha)\vee (\bigwedge_{\beta\in\Psi^{\in-}}\beta)\big).\]
    Applying the distributive property of $\wedge$ and $\vee$, we get that 
    \[\Phi'=\Psi^{\not\in}\wedge \bigwedge_{\alpha\in\Psi^{\in+}, \beta\in\Psi^{\in-}}(\alpha\vee\beta),\]
    which is a formula in 2-SAT, with $\disc{\Phi'}=\pi_D(\disc{\Phi})$ as required.    
\end{proof}

\begin{lemma}\label{lem:projectionHorn}
    Let $S$ be a finite set of logical relations each equivalent to a Horn SAT (Dual-Horn SAT) formula. Then, for every $d$-dimensional instance $\Phi\in SAT_C(S)$ and every $D\subseteq [d]$, there exists a Horn SAT (Dual-Horn SAT) formula $\Phi'$ such that $\disc{\Phi'}=\pi_D(\disc{\Phi})$.
\end{lemma}
\begin{proof}
    We only run the proof for Horn SAT, the case for Dual-Horn SAT works symmetrically. Let again $D=\{x_d\}$. We perform the same steps as in the previous proof of \Cref{lem:projection2SAT}. Now, instead of being in 1-CNF, the clauses of $\Psi^{\in+}$ are disjunctions of only negative literals. Similarly, the clauses of $\Psi^{\in-}$ are disjunctions of at most one positive literal, i.e., $\Psi^{\in-}$ is in Horn SAT.
    
    We thus see that the clauses $\alpha\vee\beta$ are again in Horn SAT, and that the whole formula $\Phi'$ must be in Horn SAT.
\end{proof}

This finishes the proof of \Cref{thm:main}. Finally, we show that \Cref{thm:main} and \Cref{thm:CNFdichotomy} cannot be strengthened to a homotopy-universality dichotomy for CSPs in general:
\begin{lemma}\label{lem:HomotopyFalse}
    There exists a finite set of logical relations $S$ such that $SAT_C(S)$ is classified as \NP-complete by \Cref{thm:schaefer}, but $SAT_C(S)$ is \emph{not} homotopy-universal.
\end{lemma}
\begin{proof}
    Consider \oneinthree, the boolean CSP (with constants) given by the constraints encoding that exactly one of three literals is true. Then, the solution space of each clause is the disjoint union of (at most) three (at least) $d-3$-dimensional faces of the hypercube. Taking the intersection of such solution spaces can only yield disjoint faces. Thus, for every $\Phi$ in \oneinthree, $\induc{\disc{\Phi}}$ is the disjoint union of contractible components, and \oneinthree is not homotopy-universal.
\end{proof}

\clearpage
\bibliography{literature}

\end{document}

%% file: macros.tex
\usepackage[utf8]{inputenc}
\usepackage{amsmath,amssymb,amsfonts,mathrsfs,amsthm}
\usepackage{hyperref}

\usepackage{xspace}
\usepackage{authblk}
\usepackage{graphicx}
\usepackage{comment}
\usepackage{thm-restate}

\usepackage[capitalise]{cleveref}

\newcommand{\R}{\mathbb{R}\xspace}

\newcommand{\NP}{\ensuremath{\mathsf{NP}}\xspace}
\renewcommand{\P}{\ensuremath{\mathsf{P}}\xspace}

\newcommand{\mnev}{Mn{\"e}v\xspace}

\newtheorem{theorem}{Theorem}[section]

\newtheorem{corollary}[theorem]{Corollary}
\newtheorem{definition}[theorem]{Definition}
\newtheorem{lemma}[theorem]{Lemma}

\newtheorem{claim}[theorem]{Claim}
\newtheorem{fact}[theorem]{Fact}

\let\int\undefined
\DeclareMathOperator{\int}{int}

\DeclareMathOperator{\im}{im}

\newcommand{\isomorph}{\cong}

\newcommand{\ER}{\ensuremath{\exists\R}}

\newcommand{\oneinthree}{\textsc{1-in-3-SAT}\xspace}
\newcommand{\kpnsat}[3]{\textsc{$#2$-pos-$#3$-neg-$#1$-SAT}\xspace}
\DeclareMathOperator{\Sat}{Sat}
\newcommand{\disc}[1]{\ensuremath{\Sat(#1)}}
\newcommand{\induc}[1]{\ensuremath{I(#1)}}

\renewcommand{\phi}{\varphi}